\documentclass{article}

 \usepackage[final]{neurips_2019}

\usepackage[utf8]{inputenc} % allow utf-8 inputtitle
\usepackage[T1]{fontenc}    % use 8-bit T1 fonts
\usepackage{hyperref}       % hyperlinks
\usepackage{url}            % simple URL typesetting
\usepackage{booktabs}       % professional-quality tables
\usepackage{amsfonts}       % blackboard math symbols
\usepackage{nicefrac}       % compact symbols for 1/2, etc.
\usepackage{microtype}      % microtypography
\usepackage{amsmath}
\usepackage{amsthm}
\usepackage{color}
\usepackage{graphicx}
\usepackage{algorithm,algorithmic}
\usepackage{verbatim}
\usepackage{caption}
\usepackage{subcaption}
\usepackage{todonotes}
\usepackage{mathtools}

\theoremstyle{plain}
\newtheorem{thm}{Theorem}[section] % reset theorem numbering for each section

\newtheorem{prop}[thm]{Proposition}

\theoremstyle{definition}
\newtheorem{defn}[thm]{Definition} % definition numbers are dependent on theorem numbers

\newtheorem{remark}[thm]{Remark}

\newcommand{\defeq}{\vcentcolon=}

\newcommand{\mcF}{\mathcal{F}}

\newcommand{\mbR}{\mathbb{R}}

\newcommand{\iid}{\overset{\text{iid}}{\sim}}
\newcommand{\bH}{{\bf H}}
\newcommand{\ep}{\epsilon}
\newcommand{\De}{\Delta}
\newcommand{\RR}{\mathbb R}
\newcommand{\ta}{\theta}
\newcommand{\sa}{\sigma}

\newcommand{\twid}{\widetilde}
\newcommand{\de}{\delta}
\DeclareMathOperator{\tr}{trace}

\newcommand{\mscr}[1]{\mathcal #1}
\DeclareMathOperator{\argmin}{argmin}

\title{KNG: The K-Norm Gradient Mechanism}
\author{
  Matthew Reimherr \thanks{Research supported in part by NSF DMS 1712826, NSF SES 1853209, and the Simons Institute for the Theory of Computing at UC Berkeley.} \\
  Department of Statistics\\
  Pennsylvania State University \\
  State College, PA 16802 \\
  \texttt{mreimherr@psu.edu} \\
  \And
  Jordan Awan\\% \thanks{website + grant} \\
  Department of Statistics\\
  Pennsylvania State University \\
  State College, PA 16802 \\
  \texttt{awan@psu.edu} \\
}

\begin{document}
% \nipsfinalcopy is no longer used

\maketitle

\begin{abstract}
 This paper presents a new mechanism for producing sanitized statistical summaries that achieve {\it differential privacy}, called the {\it K-Norm Gradient} Mechanism, or KNG. This new approach maintains the strong flexibility of the exponential mechanism, while achieving the powerful utility performance of objective perturbation. KNG starts with an inherent objective function (often an empirical risk), and promotes summaries that are close to minimizing the objective by weighting according to how far the gradient of the objective function is from zero.  Working with the gradient instead of the original objective function allows for additional flexibility as one can penalize using different norms.  We show that, unlike the exponential mechanism, the noise added by KNG is asymptotically negligible compared to the statistical error for many problems. In addition to theoretical guarantees on privacy and utility, we confirm the utility of KNG empirically in the settings of linear and quantile regression through simulations. 
\end{abstract}

\section{Introduction}\label{s:intro}
The last decade has seen a tremendous increase in research activity related to data privacy \citep{Aggarwal2008general,Lane2014,Machanavajjhala2015,Dwork2017}.  This drive has been fueled by an increasing societal concern over the large amounts of data being collected by companies, governments, and scientists.  These data often contain vast amounts of personal information, for example DNA sequences, images, voice recordings, electronic health records, and internet usage patterns.  Such data allows for great scientific progress by researchers and governments, as well as increasingly curated business strategies by companies.  However, the such data also comes with increased risk for privacy breaches, placing greater pressure on institutions to prevent disclosures.  %methods and mathematical results continue to be developed to handle increasingly complex data and statistical problems.  

Currently, \emph{Differential Privacy} (DP) \citep{Dwork2006:Sensitivity} is the leading framework for formally quantifying privacy risk.
One of the most popular methods for achieving DP is the {\it Exponential Mechanism}, introduced by \citet{McSherry2007}, and used in \citep{friedman2010data,Wasserman2010,blum2013learning,Dwork2014}.  A major attribute of the exponential mechanism that contributes to its popularity is its flexibility; it can be readily adapted and incorporated into most statistical analyses.  In particular, its structure makes it amenable to a wide array of statistical and machine learning problems that are based on minimizing an objective function, so called ``$m$-estimators'' \citep[Chapter 5]{VanDerVaart2000}. Some examples where the exponential mechanism has been used include PCA \citep{Chaudhuri2013,Awan2019}, hypothesis testing \citep{Canonne2019}, maximum likelihood estimation (related to posterior sampling) \citep{Wang2015:PrivacyFree,Minami2016}, and density estimation \citep{Wasserman2010}.

%Examples of problems where is could be used: 

However, examples have arisen \citep{Wang2015:PrivacyFree,Awan2019} where the magnitude of the noise added by the exponential mechanism is substantially higher than other mechanisms.  Recently, \citet{Awan2019}, demonstrated that, in a very broad sense, the exponential mechanism adds noise that is not asymptotically negligible relative to the statistical estimation error, %whereas \citet{Smith2011} showed that there are efficient mechanisms for many problems . 
which other mechanism are able to achieve in different problems \cite[e.g.][]{Smith2011}.
In this paper we provide a new mechanism called the {\it K-Norm Gradient Mechanism}, or {\it KNG}, that retains the flexibility of the exponential mechanism, but with substantially improved utility guarantees. {  KNG provides a principled approach to developing efficient mechanisms that also perform well in practice.  Indeed the Laplace, $K$-norm, and PrivateQuantile mechanisms can all be viewed as instantiations of KNG.  Here we also use KNG to provide the first mechanism for private quantile regression that we are aware of, which we empirically show is efficient.} %In fact the KNG mechanism can be viewed as a method of choosing better objective functions for use in the exponential mechanism. 

At a high level, KNG uses a similar perspective to that of the exponential mechanism.  In particular, suppose that $\ell_n(\theta;D)$ is an objective, whose minimizer, $\hat \theta \in \mbR^d$, is the summary we aim to sanitize.  Here $D$ represents the particular database and $n$ the sample size of $D$.  The exponential mechanism aims to release $\tilde \theta_E$ based on the density 
\[ f_E(\theta) \propto \exp\{- c_0 \ell_n(\theta;D) \}, \]
where $c_0$ is a generic constant determined by the sensitivity of $\ell_n$ and the desired level of privacy.  Conceptually, the idea is to promote sanitized estimates whose utility, as measured by $\ell_n$, is close to that of $\hat \theta$.  Unfortunately, \citet{Awan2019}, showed that the magnitude of the noise added by the exponential mechanism is often of the same order as the statistical error (as a function of $n$), {  resulting in inefficient private estimators}.  KNG uses a similar perspective, but takes the gradient of $\ell_n$ and promotes $\theta$ that are close to the solution $\nabla \ell_n(\hat \theta) = 0$.  Since we work with the gradient, we also have the flexibility of choosing a desirable norm, which \citet{Awan2018:Structure} showed can be tailored to the problem at hand to achieve better utility.  The resulting mechanism produces a sanitized $\tilde \theta$ according to the density
\[
f_n(\theta) \propto \exp\{- c_0 \| \nabla \ell_n(\theta;D)\|_K\}, 
\]
where $\|\cdot \|_K$ is a general norm on $\mbR^d$ that can be chosen to accommodate the context of the problem. Here we see a connection between KNG and the \emph{$K$-norm mechanism}, introduced by \citet{Hardt2010}. The terminology is based on the idea of considering a set $K$ which is the convex hull of the \emph{sensitivity polytope} \citep{Kattis2017}, and defining $\lVert \cdot \rVert_K$ to be the norm such that the ball of radius one is $K$, i.e. $\{v\in \RR^d\mid \lVert v\rVert_K=1\}=K$. In fact every norm can be generated in this manner, so no there is no loss in generality from using this approach \citep{Awan2018:Structure}.

KNG can similarly be viewed as a modification of objective perturbation \citep{Chaudhuri2011,Kifer2012:PrivateCERM}.  There, one releases a sanitized estimate, $\tilde \theta_O$, by minimizing\footnote{In fact, objective perturbation minimizes $\ell_n(\theta;D)+ c \theta^\top \theta + \omega \theta^\top b$, where $c$ is a constant. We ignore this regularization term in this discussion for the simplicity of the illustration.}
\[
\tilde \theta_O = \argmin_{\theta \in \Theta} \left(
\ell_n(\theta;D)+ \omega \theta^\top b \right),
\]
where $b \in \mbR^d$ is a random vector with distribution drawn from the $K$-norm mechanism $f_b(x) \propto \exp\{-\lVert b\rVert_K\}$, and $\omega \in \mbR$ is a fixed constant based on the sensitivity of $\ell_n$ and the desired level of privacy\footnote{In \citet{Chaudhuri2011} and \citet{Kifer2012:PrivateCERM}, the $\ell_2$ norm is used. \citet{Awan2018:Structure} extend objective perturbation to allow for arbitrary norms.}.   Equivalently, one has that
$\nabla \ell_n(\tilde \theta_O; D) + \omega b = 0$, 
%\qquad
%\Longrightarrow
%\qquad 
which implies that 
$\tilde \theta_O = \nabla \ell_n^{-1}(-\omega b)$, 
assuming $\nabla \ell_n$ is invertible.  Using the change of variables formula, this implies that $\tilde \theta_O$ has density
\[
f_O(\theta) \propto \exp\{ -\omega^{-1} \lVert\nabla \ell_n(\theta)\rVert_K\} |\mbox{det}(\nabla^2 \ell_n(\theta) )|. 
\]
With KNG, the second derivative term $\nabla^2 \ell_n$ is not included. %, and replaces $| \cdot|$ with a more general $K$-norm. 
Furthermore, there are several technical requirements when working with objective perturbation that KNG sidesteps. In particular, the proof that objective perturbation satisfies DP requires the objective function to be strongly convex and twice differentiable almost everywhere \citep{Chaudhuri2011,Kifer2012:PrivateCERM,Awan2018:Structure}. While we assume strong convexity and a second derivative to prove a utility result in Theorem \ref{thm:utility1}, KNG  does not require either of these conditions to satisfy DP. This allows the KNG mechanism to be applied in more general situations (such as median estimation and quantile regression, explored in Section \ref{s:ex}), and requires fewer calculations to implement.

The remainder of this paper is organized as follows. In Section \ref{s:DP} we recall the necessary background on differential privacy and the exponential mechanism. In Section \ref{s:grad} we formally define KNG and show that it achieves $\ep$-DP with nearly the same flexibility as the exponential mechanism.  We also provide a general utility result that shows that the noise introduced by KNG is of order $O_p(n^{-1})$, which is negligible compared to the statistical estimation error, which is typically $O_p(n^{-1/2})$. We also show that the noise introduced by KNG is asymptotically from a $K$-norm mechanism. In section \ref{s:ex} we provide several examples of KNG applied to statistical problems, including mean estimation, linear regression, median/quantile estimation, and quantile regression. We also  illustrate the empirical advantages of KNG in the settings of linear and quantile regression through simulations.  We conclude in Section \ref{s:con} by discussing challenges and potential extensions of KNG.

\section{Differential Privacy Background}\label{s:DP}
Differential privacy (DP), introduced by \citet{Dwork2006:Sensitivity} has taken hold as the primary framework for formally quantifying privacy risk. Several versions of DP have been proposed, such as approximate DP \citep{Dwork2014}, concentrated DP \citep{Dwork2016,Bun2016}, and local DP \citep{duchi2013local}, all of which fit into the axiomatic treatment of formal privacy given by \citet{Kifer2012:Axiomatic}. In this paper, we work with pure $\ep$-DP, stated in Definition \ref{def:dp}.

Let $\mathcal D^n$ denote the collection of all possible databases with $n$ units. % We begin by defining the {\it Hamming Distance} between two databases.  
%\begin{defn}[Hamming Distance]\label{HammingDistance}
  %Let $\mathcal X$ be any set. We write the $n$-fold \emph{Cartesian Product} $\mathcal X^n = \{(X_1,\ldots, X_n)\mid X_i \in \mathcal X, 1\leq i \leq n\}$. Then 
  The bivariate function $\delta:\mathcal D^n \times \mathcal D^n \rightarrow \mathbb R$, 
  which maps $\delta(D,D') \defeq \#\{i \mid D_i \neq D'_i\}$, is called the \emph{Hamming Distance} on $\mathcal D^n$.
%\end{defn}
\noindent It is easy to verify that $\de$ is a metric on $\mathcal D^n$. If $\delta(D,D')=1$ then $D$ and $D'$ are said to be \emph{adjacent}.

Let $f:\mathcal D^n \to \Theta$ represent a summary of $\mathcal D^n$, and $\mcF$ a $\sigma$-algebra on $\Theta$, such that $(\Theta, \mcF)$ is a measurable space. 
%From a probabilistic perspective, 
A {\it privacy mechanism} is a family of probability measures $\{ \mu_D: D \in \mathcal D^n\}$ over $\Theta$.%  We can now define what we mean when we say the mechanism satisfies $\epsilon$-DP.% While DP was originally introduced in \citet{Dwork2006:Sensitivity}, over time the definition has been made more precise to account for measure theoretic technicalities, such as the version given in \citet{Wasserman2010} in terms of conditional probabilities.

\begin{defn}[Differential Privacy: \citealp{Dwork2006:Sensitivity}]\label{def:dp}
%Let $\ep>0$ be given, let $\mathcal X$ be a set, and let $(\mathcal Y, \mscr F)$ be a measurable space. 
A privacy mechanism $\{\mu_D : D \in \mscr D^n \}$  
%set of probability measures $\{ \mu_X \mid X\in \mathcal X^n\}$ on $(\mathcal Y, \mscr F)$ 
satisfies $\ep$-Differential Privacy ($\ep$-DP) if for all $B \in \mcF$ and  adjacent $D,D' \in \mscr D^n$,  
\[
\mu_D(B) \leq \mu_{D'}(B) \exp(\ep).
\]
%where if $\mu_X(B)=\mu_{X'}(B=0$, we define the ratio to be $1$.
\end{defn}

%One of the earliest mechanisms developed for differential privacy is the exponential mechanism, introduced by \citet{McSherry2007}. 
The exponential mechanism, introduced by \citet{McSherry2007} is a central tool in the design of DP mechanisms \citep{Dwork2014}. In fact every mechanism can be viewed as an instance of the exponential mechanism, by setting the objective function as the log-density of the mechanism. In practice, it is most common to  
%the exponential mechanism is often used by setting 
set the objective as a natural loss function, such as an empirical risk.

%\todo{why do the prop and def environments use different brackets?}
\begin{prop}[Exponential Mechanism: \citealp{McSherry2007}]\label{prop:ExpMech}%\label{ExponentialMechanism}
Let $(\Theta, \mscr F, \nu)$ be a measure space, and let $\{ \ell_n(\theta;D): \Theta \rightarrow \RR \mid D\in \mscr D^n\}$ be a collection of measurable functions indexed by the database $D$.  We say that this collection has a %($\mscr F/\mscr R_1$) 
 finite \emph{sensitivity} $\Delta$, if %$\Delta_\xi$ is the smallest value satisfying 
\[ %\ess_{\nu} \sup_{b\in \mathcal Y} \sup_{\substack{\de(X,X')=1\\X,X'\in \mathcal X^n}}  
|\ell_n(\theta;D)- \ell_n(\theta;D')| \leq \Delta < \infty,\]
for all adjacent $D,D'$ and $\nu$-almost all $\theta \in \Theta$. %, and if $\Delta_\xi < \infty$.  
If $\int_{\Theta} \exp(-\ell_n(\theta;D)) \ d\nu(\theta)<\infty$ for all $D\in \mathcal D$, then the collection of probability measures $\{\mu_D\mid D\in \mathcal D\}$ with densities (with respect to $\nu$) 
\[f_D(\theta) \propto \exp\left\{\left(\frac{-\ep}{2\Delta}\right)\ell_n(\theta;D)\right\}\quad
\text{satisfies $\ep$-DP.}\]
\end{prop}
Intuitively, $\ell_n(\theta;D)$ provides a score quantifying the utility of an output $\theta$ for the database $D$. We use the convention that smaller values of $\ell_n(\theta;D)$ provide more utility. % since the unsantized estimator is taken to be the minimizer.
So, the exponential mechanism places more mass near the minimizers of $\ell$, and less mass the higher the value of $\ell_n(\theta;D)$.

%%%%%%%%%%%%%%%%%%%%%%%%%%%%%%%%%%%%%%%%%%%%%%%%%%%%%%%%%%%%%%%%%%%%%%%%%%%%
%%%   KNG
%%%%%%%%%%%%%%%%%%%%%%%%%%%%%%%%%%%%%%%%%%%%%%%%%%%%%%%%%%%%%%%%%%%%%%%%%%%%%
\section{The K-Norm Gradient Mechanism}\label{s:grad}
In Section \ref{s:DP} we considered an arbitrary measure space, $(\theta,\mathcal F, \nu)$, when defining DP and the exponential mechanism. However, here we focus on $\mbR^d$.  The KNG mechanism cannot be defined to quite the generality of the exponential mechanism since we require enough structure on the parameter space to define a gradient.   Most applications focus on Euclidean spaces, so this is not a major practical concern, but there could be implications for more complicated nonlinear, discrete, or infinite dimensional settings.

\begin{thm}[$K$-Norm Gradient Mechanism (KNG)]\label{thm:kng}
Let $\Theta \subset \mathbb R^d$ be a convex set, $\lVert \cdot \rVert_K$ be a norm on $\mathbb R^d$, and $\nu$ be a $\sa$-finite measure on $\Theta$.  Let $\{ \ell_n(\theta;D): \Theta \rightarrow \mathbb R \mid D\in \mscr D^n\}$ be a collection of measurable functions, which are differentiable $\nu$ almost everywhere.  We say that this collection has %($\mscr F/\mscr R_1$) 
  \emph{sensitivity} $\Delta:\Theta\rightarrow \RR^+$, if %$\Delta_\xi$ is the smallest value satisfying 
\[ %\ess_{\nu} \sup_{b\in \mathcal Y} \sup_{\substack{\de(X,X')=1\\X,X'\in \mathcal X^n}}  
\lVert \nabla\ell_n(\theta;D)- \nabla\ell_n(\theta;D')\rVert_K \leq \Delta(\theta) < \infty,\]
for all adjacent $D,D'$ and $\nu$-almost all $\theta$. %, and if $\Delta_\xi < \infty$.  
If $\int_{\Theta} \exp(-\frac{1}{\Delta(\ta)}\lVert \nabla \ell_n(\theta;D)\rVert_K) \ d\nu(\theta)<\infty$ for all $D\in \mathcal D$, then the collection of probability measures $\{\mu_D\mid D\in \mathcal D\}$ with densities  (with respect to $\nu$) 
\[f_D(\theta) \propto \exp\left[\left(\frac{-\ep}{2\Delta(\theta)}\right)\lVert \nabla \ell_n(\theta;D)\rVert_K\right]
\quad \text{satisfies $\ep$-DP.}\]
\end{thm}
\begin{proof}
Set $\twid \ell_n(\ta;D) = \Delta(\ta)^{-1}\lVert\nabla\ell_n(\ta;D)\rVert_K$. Then $\twid \ell$ has sensitivity 1. By Proposition \ref{prop:ExpMech}, the described mechanism satisfies $\ep$-DP.
\end{proof}

One advantage of this approach over the traditional exponential mechanism is that the sensitivity calculation is often simpler (e.g. quantile regression, subsection \ref{s:QR}).  However, it also has the same intuition as the exponential mechanism.  In particular, the optimum, $\hat \theta$, occurs when $\nabla \ell_n(\hat \theta) = 0$, thus we want to promote solutions that make the gradient close to 0,  and discourage ones that make the gradient far from $0$. These concepts are closely related to $m$-estimators, $z$-estimators, and estimating equations \citep[Chapter 5]{VanDerVaart2000}.

%\begin{remark}
Since KNG utilizes the gradient, it links in nicely to optimization methods such as gradient descent.  However, it could also suffer from some of the same challenges as gradient descent.  Namely, if the objective function has multiple local minima, then KNG will promote output near each these points.  For this reason, a great deal of care should be taken with KNG when applying to non-convex objective functions, such as fitting neural networks \citep{gori1992problem}.
%\end{remark}

\subsection{Asymptotic Properties}

While flexibility of a mechanism is an important concern, ultimately the utility of the output is of primary importance. \citet{Awan2019} showed that for a large class of objective functions, the exponential mechanism introduces noise of magnitude $O_p(n^{-1/2})$, where $n$ is the sample size. For many statistical problems the non-private error rate is also $O_p(n^{-1/2})$ \citep[Chapter 5]{VanDerVaart2000}, meaning that the exponential mechanism introduces noise that is not asymptotically negligible. %{  In terms of sample sizes, this means that asymptotically the exponential mechanism requires a constant $>1$ times as many samples to achieve the same error as the non-private estimator.}

Under similar assumptions, we show in Theorem \ref{thm:utility1} that KNG has aymptotic error  $O_p(n^{-1})$, which is asymptotically negligible compared to the statistical error. In fact, Theorem \ref{thm:utility1} shows that the noise introduced is asymptotically from a $K$-norm mechanism \citep{Hardt2010,Awan2018:Structure}, which generalizes the Laplace mechanism. %This family of mechanisms have also been studied in \citet{Hall2012} in the context of density estimation. \citet{Xiao2015} work with $K$-norm mechanisms for the privatization of location data. The $\lVert \cdot \rVert_\infty$ version has seen use in \citet{Steinke2017}. The most common form of the $K$-norm mechanism is using $\lVert \cdot \rVert_1$, which results in the Laplace mechanism. 

The assumptions in Theorem \ref{thm:utility1} are chosen to capture a large class of common loss functions, which include many convex empirical risk functions and log-likelihood functions. Mathematically, the assumption that $\ell$ is twice-differentiable and strongly convex allow us to use a one term Taylor expansion of $\nabla \ell$ about $\hat \ta$, and guarantee that the integrating constants converge. The proof of Theorem \ref{thm:utility1} is found in the Supplementary Materials.

\begin{thm}[Utility of KNG]\label{thm:utility1}
{  Let $\Theta\subset \RR^d$ be a convex set, $\lVert \cdot \rVert_K$ a norm on $\RR^d$, $\nu$ a $\sigma$-finite measure om $\Theta$, and $\ell_n(\theta)\defeq \ell_n(\theta;D)$ be a sequence of objective functions which satisfy the assumptions of Theorem \ref{thm:kng}, with sensitivity $\Delta(\theta)$. 
%Assume the observed record, $D = X_1,\dots,X_n$, and corresponding
%Denote the sequence of objective functions $\ell_n(\theta):= \ell_n(\theta;D)$, for $\theta \in \Theta \subset\mathbb R^d$ \todo{Should we say something like assumptions of Thm 3.1 hold?} and $n=1,2,\dots$.  Let $\nabla \ell_n(\theta)$ and $\bH_n(\theta)$ denote the gradient and Hessian of $\ell_n$ respectively.
%Let $\lVert \cdot \rVert_K$ be any norm on $\RR^d$. 
We further assume that}
\begin{enumerate}
    \item $n^{-1}\ell_n(\theta)$ %is bounded from above (across both $X$ and $b$) and $\hat b = \arg \max_{b\in \RR^p} \xi_X(b)$ exists and is unique;
    are twice differentiable (almost everywhere) convex functions and there exists %a finite $\alpha \in \RR^+$
    a finite $\alpha>0$ 
    %such that 
    such that $n^{-1} \bH_n(\theta)$ has eigenvalues greater than $\alpha$. % ({\red We could instead assume that $\lVert \nabla \xi(b)\rVert^2$ is $n^2\alpha$-strongly convex.})
    for all $n$ and $ \theta \in \Theta$;
    %\item $n^{-1} \xi_X(b)$ is twice continuously  differentiable at $\hat b$  (uniformly in $n$, i.e. twice equidifferentiable);
    \item the minimizers satisfy $\hat \theta \to \theta^\star \in \mathbb R^d$ and $n^{-1}\bH_n(\hat \theta) \to \Sigma^{-1}$ where $\Sigma$ is a $d \times d$ positive definite matrix;
    \item %there exists a function 
    {  $\Delta(\theta)$ is continous in $\theta$, constant in $n$, and % such that $\lVert \nabla \ell_n(\theta;D) - \ell_n(\theta;D')\rVert_K \leq \Delta(\theta)$ for all $n$, $\theta\in \Theta$, and all adjacent $D$ and $D'$. Furthermore, assume that 
    there exists $\Delta>0$ such that  $\Delta\leq \Delta(\theta).$}
    %$\ell_n$ has sensitivity $\Delta(\theta)$, which is continuous in $\theta$, constant in $n$, and has a positive lower bound $\Delta\leq \Delta(\theta)$
\end{enumerate}
Assume the base measure, $\nu$, has a %strictly positive and twice continuously differentiable 
bounded, differentiable density $g(\theta)$ (with respect to Lebesgue measure) which is strictly positive in a neighborhood of $\theta^\star$. %, and define $h(b) = \log(g(b))$. 
Then the sanitized value $\tilde \theta$ drawn from the KNG mechanism with privacy parameter $\ep$ is asymptotically $K$-norm. That is, the density of $Z=n(\tilde \theta -\hat \theta )$ converges to a $K$-norm distribution, with density (wrt $\nu$) proportional to 
$\displaystyle
f(z) \propto \exp\left(\frac{-\ep}{2\Delta(\theta^*)} \lVert \Sigma^{-1} z\rVert_K\right).$
\end{thm}

{  The proof of the CLT for the exponential mechanism in \citet{Awan2019}, as well as the proof of Theorem 3.2, both rely on a Taylor expansion of the objective function. In both cases, it is assumed that the Hessian converges, when scaled by $n$, to a positive definite matrix. However, using the original objective function requires two derivatives before the Hessian appears in the Taylor expansion, whereas the use of the gradient only requires one derivative. The consequence of this is that the traditional exponential mechanism results in a quadratic numerator inside the exponent, whereas KNG has a (normed) linear numerator.  Asymptotically, this gives an $O_p(n^{-1/2})$ Gaussian noise for the exponential mechanism and an $O_p(n^{-1})$  $K$-norm noise for KNG. Geometrically, it seems that the use of an objective function which behaves linearly (in absolute value) near the optimum, rather than quadratic, results in better asymptotic utility. By using the normed-gradient, we construct an objective function with this property. }

{  The assumptions in Theorem 3.2 are very similar to the assumptions for the CLT in \citet{Awan2019}. So, whenever these properties hold, we know that KNG results in an $O_p(n^{-1})$ privacy noise whereas the exponential mechanism is $O_p(n^{-1/2})$. To further emphasize the importance of this result, we note that the magnitude of the noise introduced for privacy can have a substantial impact on the sample complexity. Asymptotically, KNG requires exactly the same sample size as the non-private estimator, whereas the exponential mechanism requires a constant $>1$ multiple of the non-private sample size to achieve the same accuracy. }

As we see in Section \ref{s:ex}, in the problem of quantile regression the assumptions of Theorem \ref{thm:utility1} do not hold, meaning that while we guarantee privacy in that setting, we can't guarantee the utility of the estimator. However, we see in Figure \ref{fig:QuantilePlot} that KNG still introduces asymptotically negligible noise, suggesting that the assumptions of Theorem \ref{thm:utility1} can likely be weakened to accomodate a larger class of objective functions.

\begin{remark}\label{rmk:ObjPert}
Based on the discussion in Section \ref{s:intro}, a result similar to \ref{thm:utility1} may hold for objective perturbation as well. The main issue is dealing with the change of variables factor $|\det \bf H_n(\theta)|$, which may or may not contribute to the asymptotic form. We suspect that when both KNG and objective perturbation are applicable (e.g. linear regression, see subsection \ref{s:linSims}), they will have similar performance. However, as KNG does not require a second derivative (or convexity), it is applicable in more settings than objective perturbation (e.g. quantile regression, see subsection \ref{s:QR}). 
\end{remark}

\section{Examples}\label{s:ex}%\todo{notation in the examples needs to be cleaned up. In general I like the notation $\theta^*$ for "true parameter" and $\theta$ for the argument.  So, I would swap $\mu \to \theta$ for the mean, I would swap $\theta \to \theta^*$ for regression and then $b \to \theta$, etc.}
%%%%%%%%%%%%%%%%%%%%%%%%%%%%%%%%%%%%%%%%%%%%%%%%%%%%%%%%%%%%%%%%%%%%%%%%
%%%   MEAN ESTIMATION
%%%%%%%%%%%%%%%%%%%%%%%%%%%%%%%%%%%%%%%%%%%%%%%%%%%%%%%%%%%%%%%%%%%%%%%%%%
\subsection{Mean Estimation}
Mean estimation is one of the simplest statistical tasks, and one of the first to be solved in DP. Assuming bounds on the data, the mean can be estimated by adding Laplace noise \citep{Dwork2006:Sensitivity}. Recently there has been some work developing statistical tools for the mean under differential privacy, such as confidence intervals in the normal model \citep{Karwa2017} and hypothesis tests for Bernouilli data \citep{Awan2018:Binomial}. We show that KNG recovers the $K$-norm mechanism when estimating the mean, a generalization of the Laplace mechanism.

Let $x_1,\ldots, x_n\in \RR^d$, which we assume are drawn from some population with mean $\theta^*$. To estimate $\theta^*$, we use the  sum of squares as our objective function:
\[
\ell_n(\theta;D) = \sum_{i=1}^n \|x_i - \theta\|_2^2
\qquad \text{and} \qquad
\nabla \ell_n(\theta;D) = -2\sum_{i=1}^n (x_i - \theta)
= -2n( \bar x - \theta).
\]
Turning to the sensitivity, if we assume that there exists a constant $r$ such that $\|x_i\|_K \leq r < \infty$ for some norm $\|\cdot \|_K$, then the sensitivity of the gradient is 
$\|\nabla \ell_n(\theta;D) - \nabla \ell_n(\theta;D')\|_K
= 2 \|x_1 - x_1'\|_K \leq 4r.$ 
Thus the mechanism becomes 
$\displaystyle
f_{n}(\theta) \propto 
\exp \left\{
- (n \epsilon/(4r))\left\| \bar x - \theta\right\|_K
\right\},$ 
which is exactly a K-norm mechanism \citep{Hardt2010}. % When the norm is $\lVert \cdot \rVert_1$, this gives the Laplace mechanism. 
So $\tilde \theta - \bar x$ has mean $0$ and standard deviation $ O_p(n^{-1})$.  Thus, the noise added for privacy is asymptotically negligible compared to the statistical error $O_p(n^{-1/2})$.

\begin{remark}
Because the KNG results in a location family in this case, the integrating constant does not depend on the data. So, we do not need to divide $\epsilon$ by 2 in the density, and may instead draw from $f_n(\theta) \propto \exp\left\{\frac{n\epsilon}{2r} \left\| \bar x- \theta \right \|_K\right\}$, which is how the $K$-norm mechanism is normally stated.
\end{remark}

%%%%%%%%%%%%%%%%%%%%%%%%%%%%%%%%%%%%%%%%%%%%%%%%%%%%%%%%%%%%%%%%%%%%%%%%%%%%%%
%%%
%%%%%%%%%%%%%%%%%%%%%%%%%%%%%%%%%%%%%%%%%%%%%%%%%%%%%%%%%%%%%%%%%%%%%%%%%%%%%%%%
\subsection{Linear Regression}
There has been a great deal of work developing DP methods for linear regression \citep{Zhang2012,Song2013,Dwork2009,Chaudhuri2011,Kifer2012:PrivateCERM,Sheffet2017}.  In this section, we detail how KNG can be used to estimate the coefficients in a linear regression model. We observe pairs of data $(x_i,y_i)$, where $y_i \in \RR$ and $x_i\in \RR^d$, which we assume are modeled as  
$y_i = x_i^\top \theta^* + e_i,$ 
where the errors are iid with mean zero and are uncorrelated with $x$. Our goal is to estimate $\theta^*$. To implement KNG, we assume that the data has been pre-processed such that $-1\leq x_{i}\leq 1$ and $-1\leq y_i \leq 1$ for all $i=1,\ldots,n$. We also assume that $\lVert \theta^* \rVert_1 \leq B$. The usual non-private estimator for $\theta^*$ is the least-squares,
%/maximum likelihood estimator
%I want to be careful saying MLE here since it is the MLE under the normal model, but we aren't stating any distributional assumptions
 which minimizes 
the objective function $\ell(\theta; D) = \sum_{i=1}^n (y_i - x_i^\top \theta)^2$. KNG requires a bound on the sensitivity of $\nabla\ell_n$:
\begin{align*}
    \lVert \nabla \ell_n(\theta;D) - \nabla \ell_n(\theta;D') \rVert %& = 2\lVert (y_1 - x_1^\top \theta) x_1  - (y_2 - x_2^\top\theta)x_2\rVert\\
    &\leq \sup_{y_1,x_1,\theta}4 \lVert (y_1 - x_1^\top \theta)x_1 \rVert
    %&\leq \sup_{x_1,\theta} 4\lVert (1 + \lVert \theta \rVert_1)x_1 \rVert\\
    =\sup_{x_1}4(1 + B) \lVert x_1 \rVert.
\end{align*}
By using the $\ell_\infty$ norm, we get the tightest bound, since $\lVert x_1 \rVert_\infty\leq 1$. KNG samples from the density 
\begin{equation}\label{eq:LR}
    f_n(\theta)\propto \exp \left( \frac{-\ep}{8(1+B)} \left\lVert \sum_{i=1}^n (y_i - x_i^\top\theta)x_1^\top\right\rVert_\infty\right),
\end{equation}
    with respect to the uniform measure on $\Theta = \{\theta \mid \lVert \ta\rVert_1\leq B\}$.

\begin{remark}
Alternative sensitivity bounds can be obtained by choosing other bounds on $x$ and $y$. The bound on $\theta^*$ can be removed entirely, allowing $\Delta$ to depend on $\theta$. In that case, a nontrivial base measure will be required as the resulting density is not integrable with respect to Lebesgue measure.  
We prefer to use the given sensitivity bound as it allows a fairer comparison against the exponential mechanism and objective perturbation in subsection \ref{s:linSims}.
\end{remark}

%\todo{ Fig 1 and 2: You might increase the font and put the legends in the "top right" (R has a direct command for this).  Fig 1: I think we should make a third plot where we zoom in 100 to 5000 or so because we beat Obj Pert there, but not exp mech, which is sort of interesting.}
\begin{minipage}{.495\textwidth}
\includegraphics[width=\textwidth]{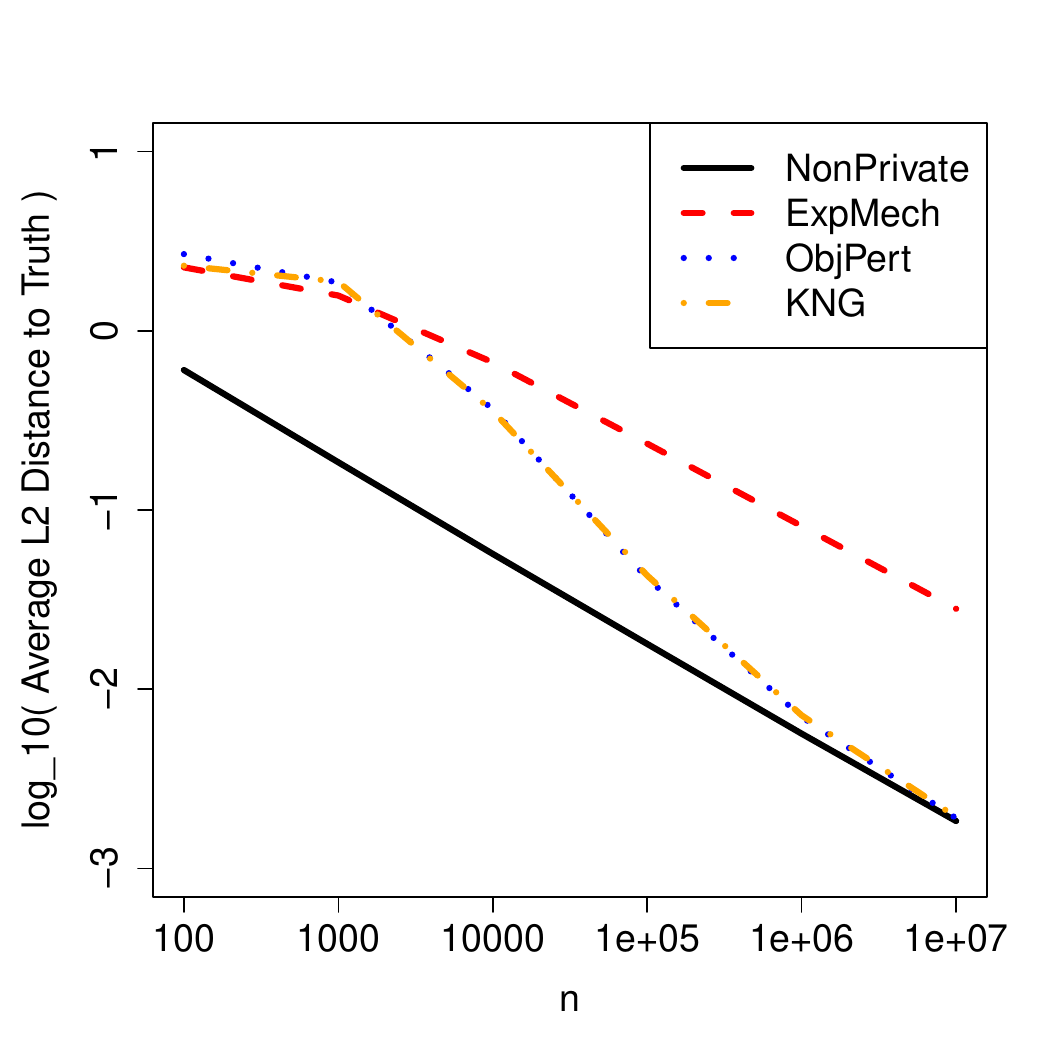}
%\fig{width=\textwidth}{KNGplotLinear1.pdf}
\captionof{figure}{Simulation comparing the non-private MLE, exponential mechanism, objective perturbation, and KNG for linear regression.}
\label{fig:LinearPlot}
\end{minipage}
\hspace{.1cm}
\begin{minipage}{.495\textwidth}
%\fig{width=\textwidth}{KNGplotQm1.pdf}{Quantile Regression}
\includegraphics[width=\textwidth]{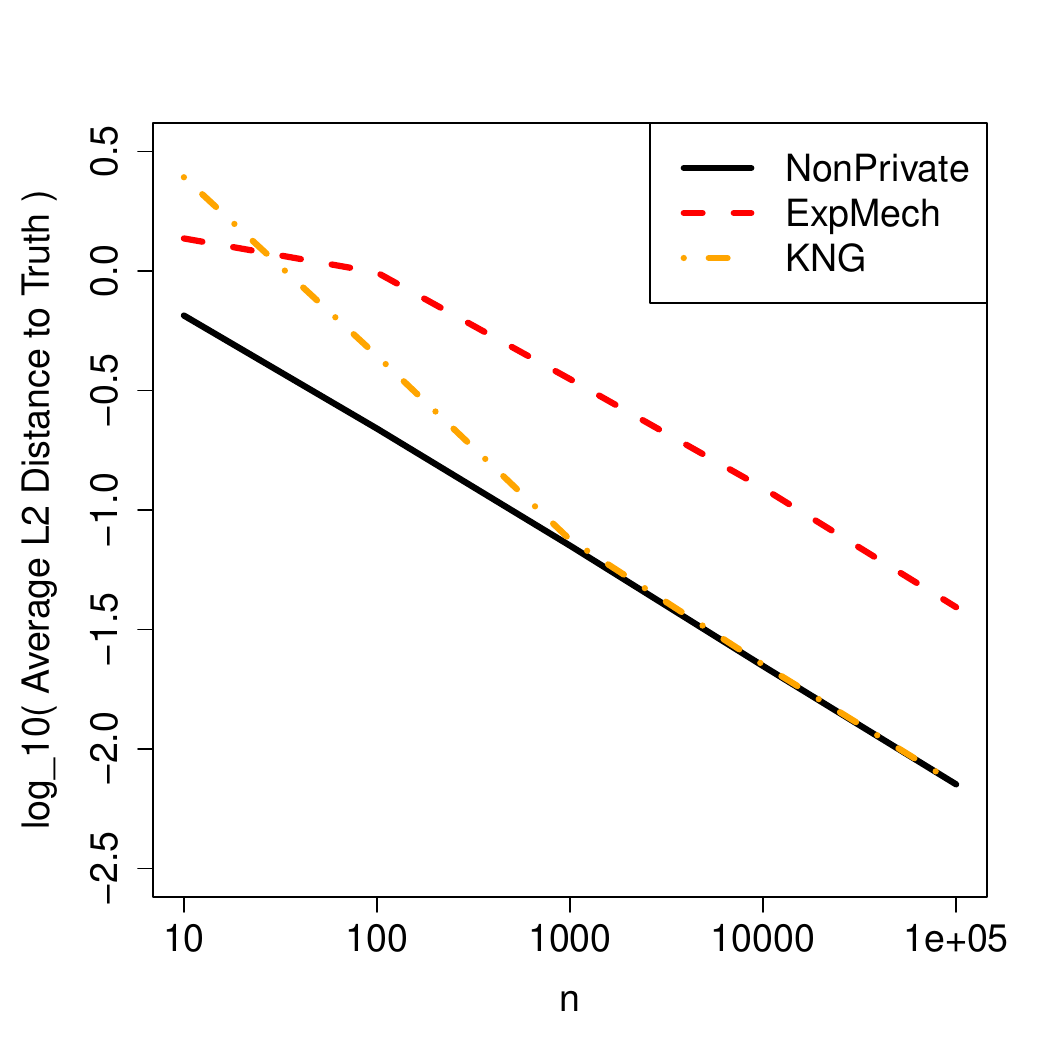}
\captionof{figure}{Simulation comparing the non-private, exponential mechanism, and KNG for quantile regression.}
\label{fig:QuantilePlot}
\end{minipage}

%%%%%%%%%%%%%%%%%%%%%%%%%%%%%%%%%%%%%%%%%%%%%%%%%%%%%%%%%%%%%%%%%%%%%%%%%%%%%%%
%%%   LINEAR REGRESSION SIMULATIONS
%%%%%%%%%%%%%%%%%%%%%%%%%%%%%%%%%%%%%%%%%%%%%%%%%%%%%%%%%%%%%%%%%%%%%%%%%%%%%%%%
\subsection{Linear Regression Simulation} \label{s:linSims}

In this section, we examine the finite sample performance of the KNG mechanism on linear regression compared to the exponential mechanism and objective perturbation mechanism. KNG samples from the density \eqref{eq:LR}, the exponential mechanism samples from \[f_n(\theta)\propto \exp \left( \frac{-\ep}{2(1 + B)^2} \sum_{i=1}^n (y_i - x_i^\top \theta)^2\right),\]
and objective perturbation  draws a random vector $b$ from the density $f(b) \propto \exp\left(- \frac{\ep}{8(1+B)} \lVert b\rVert_\infty\right)$, and then finds the optimum of the modified objective: 
$\arg\min_{\lVert\theta\lVert_1\leq 1} \ell_n(\theta;D) + \frac{\gamma}{2} \theta^\top \theta + \theta^\top b,$ 
where $\gamma= (\exp(\ep/2)-1)^{-1}(2d)$ and $d$ is the dimension of the $x_i$'s. For all three mechanisms we assume the bound on $\lVert \theta^* \rVert_1$ is $B=1$. Details on these mechanisms for linear regression can be found in the Supplementary Materials. %The approach is similar to that in \citet{Awan2018:Structure}. \todo{I would put citations here for exp + obj where similar approaches can be found.} 

For the simulations the true regression vector $\theta^* \in \mathbb R^{12}$ is $\theta^* = (0,-1,-1+2/11,-1+4/11,\ldots,1-2/11)$, and so $d=12$. For each $n$ in $10^2,10^3,10^4,\ldots,10^7$ we run 100 replicates of Algorithm \ref{SimulationRegression} at $\epsilon=1$.  For KNG and exponential mechanism, we draw samples using a one-at-a-time MCMC procedure with 10000 steps.

At the end, we compute the average distance over the 100 replicates for each mechanism and for each sample size $n$. The results are plotted in Figure \ref{fig:LinearPlot}, taking the base 10 log of both axes.  At each $n$ value and for each mechanism, the Monte Carlo standard errors are between $0.01380$ and $0.02729$, in terms of the log-scale used in the plot. The benefit of plotting in this fashion is that it makes it easier to understand the asymptotic behavior of each estimator. 

Since we know that the estimation error of the non-private MLE is $\mathrm{error}=Cn^{-1/2}$, taking the log of both sides shows that the convergence should appear as a straight line with slope $-1/2$:
$\log(\mathrm{error}) =-\frac 12 \log(n) + \log(C),$ which is the black line in Figure \ref{fig:LinearPlot}.

As \citet{Awan2019} showed, the asymptotic estimation error of the exponential mechanism is $\mathrm{error}=Kn^{-1/2}$, where $K$ is a constant greater than $C$. Taking the log of both sides gives another line with slope $-1/2$, but with a higher intercept:
$\log(\mathrm{error}) = -\frac 12 \log(n) + \log(K),$ which we see in red in Figure \ref{fig:LinearPlot}.

On the other hand, for KNG and objective perturbation (based on Remark \ref{rmk:ObjPert}) , the asymptotic estimation error is $\mathrm{error}=Cn^{-1/2}+Kn^{-1}$, which when logged shows that for larger $n$, the curve approaches the line of the non-private estimation error from above:
$\log(\mathrm{error}) = -\frac 12 \log(n) + \log(C+Kn^{-1/2}),$ which is also confirmed in Figure \ref{fig:LinearPlot}.

\begin{algorithm}
\caption{Regression Simulation}
\scriptsize
INPUT: $n$, $\epsilon$, $d$, $\theta^*$.
\begin{algorithmic}[1]
  \setlength\itemsep{0em}
  \STATE  Generate $X\in \mathbb R^{n\times{d}}$ such that $X_{i,1}=1$ and $X_{ij} \iid U(-1,1)$ for $i=1,\ldots, n$ and $j=2,\ldots,d$.
\STATE Generate independent errors $e_i \sim N(0,1)$ for $i=1,\ldots, n$.
\STATE Compute the responses $Y_i = X_i\theta^* + e_i$.
\STATE Set $R = \max_{i} |Y_i|$.
\STATE Set $Y'_i = Y_i/R$.
\STATE Use $X$ and $Y'$ to estimate the regression coefficient via the non-private estimator, and each DP mechanism.
%\begin{itemize}
%\item The Non-private estimator
%\item Exponential Mechanism 
%\item KNG
%\end{itemize}
\STATE Multiply the estimates by $R$ to estimate $\theta^*$.
\STATE Compute the euclidean distance between the estimate and the true $\theta^*$ for each estimator.
\end{algorithmic}
OUTPUT: Average distances of the estimates to the true $\theta^*$.
\label{SimulationRegression}
\end{algorithm}

%%%%%%%%%%%%%%%%%%%%%%%%%%%%%%%%%%%%%%%%%%%%%%%%%%%%%%%%%%%%%%%%%%%%%%%%%%%%%
%%%   MEDIAN ESTIMATION
%%%%%%%%%%%%%%%%%%%%%%%%%%%%%%%%%%%%%%%%%%%%%%%%%%%%%%%%%%%%%%%%%%%%%%%%%%%%%%
\subsection{Median Estimation}
%\todo{if we need space, we could cut the median estimation and just say its a special case of quantile regression}
Just as in the mean estimation problem, we observe $D = (x_1,\ldots, x_n)$, where $x_i \in \RR^d$, and our goal is to estimate the population median. In the case when $d=1$, the median can be estimated using the empirical risk function  
$\ell_n(\theta;D) = \sum_{i=1}^n |x_i -\theta|.$   
%
%
\begin{comment}
It may be concerning that this objective is not differentiable everywhere, however, KNG only requires that the gradient exist on a set of measure one.  
%Let $\ell_n$ be the same as for mean estimation but replace squares with absolute values.  
It is convenient to note that $|x_i - \theta|$ is $x_i - \theta$ if $x_i > \theta$ and $\theta - x_i$ if $x_i \leq \theta$. 

Then we have for $\theta \neq x_i$,
\[
\nabla  \ell_n(\theta;D) = -  \sum_{i=1}^n I(x_i > \theta) + \sum_{i=1}^n I(x_i \leq \theta)
= -  n (1 - \hat F(\theta; D)) + n \hat F(\theta;D)
= - n (1 - 2 \hat F(\theta;D)).
\]
where $\hat F(\theta; D) = \frac 1n \sum_{i=1}^n I(x_i \leq \theta)$ is the empirical CDF of the $\{x_i\}$. Now notice that the sensitivity is bounded by 2 since moving an $x_i$ above $\theta$ to below $\theta$ (or vice versa) will change each count by 1, while leaving it above $\theta$ leaves the count the same.  So the mechanism becomes
\[
f_{n} (\theta) \propto \exp \left\{
 -\frac{\epsilon n}{4} \left | 1 - 2 \hat F(\theta;D)\right|
\right\}.
\]
In fact, we have recovered the \emph{PrivateQuantile} algorithm proposed by \citet{Smith2011}, in the case of the median. \citet{Smith2011} shows that \emph{PrivateQuantile} has good utility guarantees; this exercise demonstrates that KNG gives a unified method of developing mechanisms with desirable properties. 
\end{comment}
%
In general for $d \geq 1$, we are estimating the \emph{geometric median} \citep{minsker2015geometric}, which can be expressed as $\arg\min_m \mathbb E\lVert X-m\rVert$, and typically the euclidean norm is used. Now, our objective becomes $\ell_n(\theta;D) =\sum_{i=1}^n \lVert x_i - \theta\rVert$. It may be concerning that this objective is not differentiable everywhere, however, KNG only requires that the gradient exist on a set of measure one.   The gradient of $\|x_i - \theta\|$ in our norm's topology is given by $d(\theta,x_i):=\|x_i - \theta \|^{-1} (x_i - \theta)$, provided that $\theta\neq x_i$.  Notice that this gives a direction in $\mbR^d$ since $\|d(\theta,x_i)\| = 1$. Using the triangle inequality, we see that the sensitivity of the gradient is bounded by 2.
%\[
%\| \nabla \ell_n(\theta; D) - \nabla \ell_n(\theta;D')\|
%%%%%= \| d_n(\theta)- d_n(\theta)'\| 
%\leq 2.
%\]
%\[
So the KNG mechanism for the median can be expressed as
\[f_{n}(\theta) \propto \exp \left\{ - \frac{\epsilon n}{4} \left\| \frac 1n\sum_{i=1}^n d(\theta,x_i)\right\|%\overline{d(\theta)}\|
\right\}.
\]
%where  $\overline{d(\theta)} = \frac 1n \sum_{i=1}^n d_i(\theta)$. 
Again, the error introduced is  $O_p(n^{-1})$, which  is negligible compared to the statistical error.%sample median's error: $O(n^{-1/2})$.

\begin{comment}
In the case where $d=1$, the  mechanism becomes
\[
f_{n} (\theta) \propto \exp \left\{
 -\frac{\epsilon n}{4} \left | 1 - 2 \hat F(\theta;D)\right|
\right\}.
\]
In fact, we have recovered the \emph{PrivateQuantile} algorithm proposed by \citet{Smith2011}, in the case of the median. \citet{Smith2011} shows that \emph{PrivateQuantile} has good utility guarantees; this exercise demonstrates that KNG gives a unified method of developing mechanisms with desirable properties. 
\end{comment}
%%%%%%%%%%%%%%%%%%%%%%%%%%%%%%%%%%%%%%%%%%%%%%%%%%%%%%%%%%%%%%%%%%%%%%%%%
%%%  QUANTILE REGRESSION
%%%%%%%%%%%%%%%%%%%%%%%%%%%%%%%%%%%%%%%%%%%%%%%%%%%%%%%%%%%%%%%%%%%%%%%%%
\subsection{Quantile Regression}\label{s:QR}
For quantile regression as for linear regression, we observe pairs of data $(x_i,y_i)$, where $y_i \in \RR$ and $x_i\in \RR^d$. 
 We assume that $Q_{Y_i\mid X_i}(\tau) = X_i^\top\theta^*_\tau,$ 
for all $i=1,\ldots, n$, where $Q_{Y\mid X}(\tau)$ is the conditional quantile function of $Y$ given $X$ for $0<\tau<1$, and $\theta^* \in \RR^p$ \citep{hao2007quantile}. For a given $\tau$, $\theta^*_\tau$ can be estimated as $\hat \theta_\tau = \arg\min_\theta \sum_{i=1}^n \rho_\tau(y_i - x_i^\top\theta),$ 
where $\rho_\tau(z) = (\tau-1)zI(z\leq 0) + \tau zI(z>0)$ is called the \emph{tiled absolute value function} \citep{koenker2001quantile}.  So, our objective function is 
\[\ell_n(\theta;D) = (\tau-1)\sum_{y_i\leq x_i^\top \theta} (y_i - x_i^\top \theta) + \tau \sum_{y_i>x_i^\top\theta} (y_i - x_i^\top \theta),\]
with gradient (almost everywhere) %$\nabla \ell_n(\theta;D) = -\tau \sum_{i=1}^n x_i + \sum_{y_i \leq x_i^\top \theta} x_i.$  
\[\nabla \ell_n(\theta;D) = (\tau-1) \sum_{y_i\leq x_i^\top \theta} (-x_i) + \tau\sum_{y_i >x_i^\top \theta} (-x_i)
=-\tau \sum_{i=1}^n x_i + \sum_{y_i \leq x_i^\top \theta} x_i.\]
We bound the sensitivity as $\Delta = 2(1-\tau)C_X$, where $\sup_{x_1} \lVert x_1\rVert \leq C_X$. Then KNG samples from  
\begin{equation}\label{eq:QR}
    f_n(\theta)\propto \exp\left\{ \frac{-\ep n}{4(1-\tau)C_X} \left\lVert -\tau \frac 1n\sum_{i=1}^n x_i + \frac 1n\sum_{y_i \leq x_i^\top \theta} x_i\right\rVert \right\}.
\end{equation}
We see a few nice benefits of the KNG method in this example. If we were to use $\ell_n$ directly in the exponential mechanism, then not only would we expect worse asymptotic performance (as demonstrated in subsection \ref{s:QRsim}), but we see that the sensitivity calculation for the gradient only requires a bound on $X$, whereas the sensitivity of $\ell_n$ requires bounds on $Y$, $X$, and $\theta^*$. Furthermore, the objective perturbation mechanism cannot be used in this setting, because $\ell$ is not strongly convex, whereas the proofs for objective perturbation \citep{Chaudhuri2009,Chaudhuri2011,Kifer2012:PrivateCERM,Awan2018:Structure} all require strong convexity. In fact, the Hessian of $\ell_n$ is zero almost everywhere making the objective perturbation inapplicable.
 
Finally note that if we are only interested in estimating the $\tau^{th}$ quantile of a set of real numbers $Y_1,\ldots, Y_n$, we could set $X_i = 1$ for all $i=1,\ldots, n$, in which case KNG samples from 
\begin{equation}
f_n(\theta)\propto\exp \left\{ \frac{-\ep n}{4(1-\tau)}\left| \tau  - \hat F(\theta;Y)\right|\right\}.
\end{equation}
In fact, this is the \emph{Private Quantile} algorithm proposed by \citet{Smith2011}, who also establish strong utility guarantees for the algorithm; this exercise demonstrates that KNG could provide, or at least contribute to, a more unified framework for developing efficient privacy mechanisms.

\subsubsection{Quantile Regression Simulation}\label{s:QRsim}

In this section, we examine the empirical performance of the KNG mechanism on quantile regression compared to the exponential mechanism. KNG samples from the density \eqref{eq:QR} using the $\lVert \cdot \rVert_\infty$ norm and setting $C_X=1$, and the exponential mechanism samples from 
\[f_n(\theta) \propto \exp\left\{ \frac{-\ep}{4\max\{\tau,1-\tau\}(1+B)}  \ell_n(\theta;D)\right\}.\]
We assume, as in subsection \ref{s:linSims} that $B=1$. Details on the exponential mechanism can be found in the Supplementary Materials. Note that objective perturbation cannot be used in this setting, as discussed in subsection \ref{s:QR}.

For the simulations, we use $\tau=1/2$ and the true regression vector $\theta_{1/2}^* \in \mathbb R^{2}$ is $\theta_{1/2}^* = (0,-1)$. For each $n$ in $10^1,10^2,\ldots,10^5$ we run 100 replicates of Algorithm \ref{SimulationRegression} at $\epsilon=1$. Samples from KNG and the exponential mechanism are obtained using 1000 steps of a one-at-a-time MCMC algorithm. At the end, we compute the average distance over the 100 replicates for each estimator and for each sample size $n$. The results are plotted in Figure \ref{fig:LinearPlot}, taking the base 10 log of both axes. At each $n$ value and for each mechanism, the monte carlo standard errors are between $0.04403$ and $0.06028$, in terms of the log-scale.% used in the plot. %The benefit of plotting in this fashion is that it makes it easier to understand the asymptotic behavior of each estimator. 

We see in figure \ref{fig:QuantilePlot} that the non-private estimate appears as a straight line with slope $-1/2$, reflecting the fact that its estimation error is $O_p(n^{-1/2})$. We also see that the exponential mechanism approaches a line with slope $-1/2$, but with a higher intercept, reflecting that it has increased asymptotic variance. Last, we see that the error of KNG approaches the error line of the non-private estimator, suggesting that KNG has the same asymptotic rate as the non-private estimator.

While the utility guarantees of Theorem \ref{thm:utility1} do not apply in this setting, as the objective function is not strongly convex, the santized estimates still achieve $\epsilon$-DP and we see from Figure \ref{fig:QuantilePlot} that, empirically, KNG introduces $o_p(n^{-1/2})$ error in this setting as well. This suggests that the assumptions in Theorem \ref{thm:utility1} can likely be weakened, {  and KNG in fact produces efficient mechanisms for an even broader set of problems than Theorem \ref{thm:utility1} prescribes.}

\section{Conclusions} \label{s:con}
In this paper we presented a new privacy mechanism, KNG, that maintains much of the flexbility of the exponential mechanism, while having substantially better utility guarantees.  These guarantees are similar to those provided by objective perturbation, but privacy can be achieved with far fewer structural assumptions.  A major draw back of the mechanism is the same as for gradient descent, which can have trouble with local minima or saddle points.  Two interesting open questions concern the finite sample efficiency of KNG vs objective perturbation and if KNG can be adapted or combined with other methods to better handle multiple minima.  

We also believe that KNG has a great deal of potential for handling infinite dimensional and nonlinear problems.  For example, parameter spaces consisting of Hilbert spaces or Riemannian manifolds have structures that allow for the computation of gradients, and which might be amenable to KNG.   With Riemannian manifolds, the gradient is often viewed as a linear mapping over tangent spaces, while in Hilbert spaces, the gradient is often treated as a linear functional.  A major advantage of KNG over other mechanisms is the direct incorporation of a general $K$-norm.  \cite{Awan2019} showed that the exponential mechanism has major problems over function spaces, which are of interest in nonparametric statistics.  These issues could potentially be alleviated by KNG with a careful choice of norm.  Many interesting challenges remain in data privacy, especially if there is additional complicated structure in the parameters or data.

% From MATT: I don't like this paragraph, I think it gives too much credit to the exponential mechanism.
%{  As KNG is a special case of the exponential mechanism, one may ask if it is warranted to call it its own ``mechanism.'' In fact, many widely used DP mechanisms such as Laplace, geometric, staircase, $K$-norm, and posterior sampling mechanism can also be viewed as an instance of the exponential mechanism. As we have shown in this paper, the KNG approach is applicable in a wide variety of situations, and offers improved utility ver a classic implementation of the exponential mechanism. For these reasons, we think it is justified to refer to KNG as a mechanism in its own right.}

{  KNG has strong connections with prior DP mechanisms, especially the exponential mechanism and objective perturbation.  Indeed, like nearly every privacy mechanism, KNG can be phrased as very particular type of exponential mechanism, however this doesn't provide insight into why KNG achieves better statistical properties.  In particular, a key point is to consider the objective function that motivated the original statistical summary, which, when used with KNG produces sanitized estimators with better statistical performance than the classic implementation of the exponential mechanism.}

{  One downside of KNG is the issue of sampling, which is similar to the exponential mechanism in that sampling from these distributions is, in general, non-trivial. We show that for mean and quantile estimation, KNG results in distributions that are efficiently sampled. However, for linear and quantile regression, we used a one-at-a-time MCMC procedure (also used for exponential mechanism). Just like sampling from an posterior distribution, developing a convenient sampling scheme is case-by-case, but often a simple MCMC procedure works well in practice. }

\section*{Acknowledgements}
This research was supported in part by NSF DMS 1712826, NSF SES 1853209, and NSF SES-153443 to The Pennsylvania State University. The first author is also grateful for the hospitality of the Simons Institute for the Theory of Computing at UC Berkeley.

\bibliographystyle{plainnat} 
%\bibliography{../../../Bibliography/DataPrivacyBib}{} 
\bibliography{./DataPrivacyBib}{}

\section{Proofs}
\begin{proof}[Proof of Theorem 3.2.]
For notational simplicity, we assume that the base measure, $\mu$, is Lebesgue.    
The density of the KNG mechanism can then be expressed as
\[
f_n(\theta) = c_n^{-1} \exp \left \{\frac {-\ep} {2\Delta(\hat \theta+z/n) } \lVert \nabla \ell_n(\theta)\rVert_K\right \} ,
\]
where $c_n$ is the normalizing constant.  
Define the random variable $Z = n(\tilde \theta - \hat \theta)$, then its density is given by
\[
f_n(z) = c_n^{-1} n^{-1} \exp \left \{ \frac{-\ep}{2\Delta(\hat \theta+z/n)}\lVert \nabla \ell_n(\hat \theta + z/n)\rVert_K\right \}.
\]
%We now aim to show that, for $z$ fixed, the density converges to a multivariate normal.  
Using a one term Taylor expansion, we have by Assumption (2) and (3) that
\begin{align*}
 \nabla \ell_n(\hat \theta + z/n) 
&=    \nabla\ell_n(\hat \theta) + \bH_n(\hat \theta)z / n+o_p(1)\\
&= \bH_n(\hat \theta)z / n+o_p(1),
%&\phantom{=}+ z^\top \xi_X''(\hat b) z / 2n^2] + o(1).
\end{align*}
%The first term will be absorbed into the constants, since it does not depend on $z$, while the second term is zero for $n$ large, leaving only the third term to contribute to the form of the density.  
where $\bH_n(\theta)$ is the Hessian matrix of $\ell_n$ evaluated at $\theta$.
%Obviously $|g(\hat b + z/n) - g(b^\star)| \to 0$, and $|\Delta(\hat b+z/n)-\Delta(b^*)|\to 0$ so the only remaining task is to show that the combined constants behave appropriately.  
Recall that
\begin{align*}
 c_n n%\exp\left\{ \frac{\epsilon}{2\Delta(\hat b)}\lVert \nabla\xi_n(\hat b)\rVert\right\}  \\
 &=  \int \exp \left \{ \frac{\ep}{2\De(\hat \theta + z/n)}\left(-\lVert\nabla\ell_n(\hat \theta + z/n)\rVert_K\right)\right \} \ dz.
\end{align*}
By Assumption (1), $\ell_n$ is strongly convex and thus % {\red Need some assumption that $\lVert \nabla \xi()\rVert$ is $\alpha$-strongly convex. Interpretation?} we have that
\[\frac{1}{\Delta(\hat \theta + z/n)}\left\langle \nabla \ell_n(\hat \theta + z/n)-\nabla \ell_n(\hat \theta), z/n\right\rangle \geq \frac{n\alpha}{\Delta} \lVert z/n\rVert_2^2.\]
Combining the Cauchy-Schwartz inequality with the fact that $\nabla\ell_n(\hat \theta)=0$ implies 
\[\frac{1}{\Delta(\hat \theta + z/n)} \lVert \nabla \ell_n(\hat \theta + z/n)\rVert_2 \geq \frac{n\alpha}{\Delta} \lVert z/n\rVert_2.\]
%\[\frac{1}{\Delta(\hat \ta + z/n)}\left\langle \nabla \ell_n(\hat \theta + z/n)-\nabla \ell_n(\hat \theta), z/n\right\rangle \geq \frac{n\alpha}{\Delta} \lVert z/n\\lVert \nabla \ell_n(\hat \theta + z/n)\rVert_2 \geq n\alpha \lVert z/n\rVert_2.\]
By the equivalence of norms on $\mathbb R^d$, we have that 
\[\frac{-1}{\Delta(\hat \theta + z/n)}\lVert \nabla \ell_n(\hat \theta + z/n) \rVert_K\leq \frac{-C\alpha}{\Delta}\lVert z \rVert_2,\]
for some constant $C$.
%\[
%\lVert \nabla\xi_X(\hat b +z/n)\rVert- \lVert \nabla\xi_n(\hat b)\rVert
%\leq   - \alpha\| z\|.
%\]
Since $\exp\{-\|z\|_2\}$ is integrable, we can apply the dominated convergence theorem to conclude that the constants converge to a nonzero and finite quantity. Since $\Delta(\theta)$ is continuous in $\theta$, we also have that $\Delta(\hat \theta + z/n) \rightarrow \Delta(\theta^*)$. Putting everything together, we can conclude that
\begin{align*}
f_n(z) \to f(z) \propto \exp\left \{\frac{-\ep}{2\De(\theta^*)}\lVert \Sigma^{-1} z\rVert_K\right\},
\end{align*}
%which, is the density of the multivariate normal.  
which is the density of the $K$-norm mechanism. 
Applying Scheffe's Theorem, we thus have both convergence in distribution as well as convergence in total variation to a $K$-norm mechanism
%\[
%n(\tilde b - \hat b)
%\overset D\to \mathrm{KNorm}\left(\lVert \cdot \rVert, \mu=0,\frac{2\Delta(b^*)\Sigma}{\ep}\right)
%\overset{D}{\to} N_p\left(0, \frac\ep{2\De}\Sigma\right).  
%\qedhere\]
\end{proof}

\section{Linear Regression}
%%%%%%%%%%%%%%%%%%%%%%%%%%%%%%%%%%%%%
\subsection{Exponential Mechanism}
Our objective function is $\ell(\theta; D) = \sum_{i=1}^n (y_i - x_i^\top \theta)^2$. For the exponential mechanism, we need to bound the sensitivity of  $\ell(\theta)$:
\begin{align*}
    |\ell_n(\theta;D) - \ell_n(\theta;D')|&= |(y_1 - x_1^\top\theta)^2 - (y_2 - x_2^\top\theta)^2|\\
    &\leq\sup_{y_1,x_1,\theta} (y_1 - x_1^\top\theta)^2\\
    &\leq \sup_{x_1,\theta} (1 + |x_1^\top\theta|)^2\\
    &\leq \sup_{\theta} (1+ \lVert \theta\rVert_1)^2\\
    & = (1+B)^2,
\end{align*}
where we used the assumptions that $\lVert x_1\rVert_\infty\leq 1$, $|y_1|\leq 1$, and $\lVert \theta^*\rVert_1\leq B$. 
The exponential mechanism with objective function $\ell(\theta)$  draws $\theta$ from 
\[f_n(\theta)\propto \exp \left( \frac{-\ep}{2(1 + B)^2} \sum_{i=1}^n (y_i - x_i^\top \theta)^2\right),\]
with respect to the uniform measure on $\{\theta \mid \lVert \theta\rVert_1\leq B\}$. 

%%%%%%%%%%%%%%%%%%%%%%%%%%%%%%%%%%%%%
\subsection{Objective Perturbation}

For objective perturbation, we use the version stated in \citet{Awan2018:Structure}, which allows us to use the same bound on the gradient as developed in subsection 4.2. Objective perturbation also requires a bound on the eigenvalues of the hessian for one datapoint:
\begin{align*}
    \text{max eigenvalue} (2x_1 x_1^\top)&\leq \tr(2x_1 x_1^\top)\\
    &=2\tr(x_1^\top x_1)\\
    &\leq 2\sum_{j=1}^d |x_{1j}|^2\\
    &\leq 2d.
\end{align*}

Objective perturbation then draws a random vector $b$ from the density $f(b) \propto \exp\left(- \frac{\ep}{8(1+B)} \lVert b\rVert_\infty\right)$ (a simple sampling algorithm for $f(b)$ is stated in \citet{Awan2018:Structure}), and then finds the optimum of the modified objective:
\[\arg\min_{\lVert\theta\lVert_1\leq 1} \ell_n(\theta;D) + \frac{\gamma}{2} \theta^\top \theta + \theta^\top b,\]
where $\gamma= \frac{2d}{\exp(\ep/2)-1}$. Since $\ell$ is convex, this new objective is also convex. We restrict the search space of $\theta$ to $\{\theta \mid \lVert \theta \rVert_1\leq B\}$, since we assume that $\lVert \theta^* \rVert\leq B$ for our bounds. 

\section{Quantile Regression}
\subsection{Exponential Mechanism}
Our objective function is $\ell_n(\theta;D) = \sum_{i=1}^n \rho_\tau(y_i - x_i^\top \beta)$. For the exponential mechanism we need to assume additional bounds on the data as well as on $\theta^*$. As in the linear regression case, we assume that $-1\leq y_i \leq 1$, $-1\leq x_i\leq 1$, and $\lVert\theta \rVert_1\leq B$. We bound the sensitivity of $\ell_n$ as
\begin{align*}
    \left| \ell_n(\theta;D) - \ell_n(\theta;D')\right|&\leq 2\left| \ell_n(\theta;D)\right|\\
    &\leq\sup_{y_1,x_1,\theta}2\max\{\tau,1-\tau\} \left|y_1 - x_1^\top \theta\right|\\
    &\leq 2\max\{\tau,1-\tau\} (1+B).
\end{align*}

The exponential mechanism then samples from the density
\[f_n(\theta) \propto \exp\left\{ \frac{-\ep}{4\max\{\tau,1-\tau\}(1+B)}  \ell_n(\theta;D)\right\}.\]

\end{document}